\def\lncs{0}

\documentclass[11pt]{article}
\usepackage[colorlinks]{hyperref}
\usepackage{xcolor}
\hypersetup{linkcolor=black,filecolor=black,citecolor=black,urlcolor=blue}
\usepackage{xspace}
\usepackage{fullpage}
\usepackage{boxedminipage}
\usepackage[boxed]{algorithm}
\usepackage{epigraph}
\usepackage{mathtools}

\newcommand{\remove}[1]{}

\usepackage{amsthm,amsmath,amssymb,xspace}

\ifnum\lncs=0
\newtheorem{proposition}{Proposition}
\newtheorem{theorem}{Theorem}
\newtheorem{definition}{Definition}

\newtheorem{lemma}{Lemma}

\newtheorem{claim}{Claim}
\newtheorem{remark}{Remark}

\newtheorem{corollary}{Corollary}

\fi
\newtheorem{myclaim}{Claim}

\usepackage{tikz}
\usepackage{relsize}
\usepackage{ctable}
\usepackage{cleveref,aliascnt}
\usepackage{xspace}

\crefname{claim}{Claim}{Claims}
\newcommand{\secparam}{n}

\ifnum\lncs=0
\newcommand{\nir}[1]{$\ll$\textsf{\color{red} Nir: { #1}}$\gg$}
\newcommand{\ilan}[1]{$\ll$\textsf{\color{blue} Ilan: { #1}}$\gg$}
\newcommand{\iftach}[1]{$\ll$\textsf{\color{orange} Iftach: { #1}}$\gg$}
\newcommand{\eylon}[1]{$\ll$\textsf{\color{purple} Eylon: { #1}}$\gg$}
\else
\newcommand{\nir}[1]{}
\newcommand{\ilan}[1]{}
\newcommand{\iftach}[1]{}
\newcommand{\eylon}[1]{}
\fi

\newenvironment{boxfig}[2]{\begin{figure}[#1]\fbox{\begin{minipage}{\linewidth}
                        \vspace{0.2em}
                        \makebox[0.025\linewidth]{}
                        \begin{minipage}{0.95\linewidth}
            {{
                        #2 }}
                        \end{minipage}
                        \vspace{0.2em}
                        \end{minipage}}}{\end{figure}}

\newcommand{\pprotocol}[4]{
\begin{boxfig}{h!}{
\begin{center}
\textbf{#1}
\end{center}
    #4
\vspace{0.2em} } \caption{\label{#3} #2}
\end{boxfig}
}

\newcommand{\protocol}[4]{
\pprotocol{#1}{#2}{#3}{#4} }
\newcommand{\cH}{{\mathcal H}}

\newcommand{\cD}{{\mathcal D}}

\newcommand{\pr}[1]{\mathop{\mathbf {Pr}}\! \left[ {#1} \right]}
\newcommand{\prob}[2]{\mathop{\mathbf {Pr}}_{#1}\! \left[ #2 \right]}
\newcommand{\E}[1]{\mathop{\mathbf E}\! \left[ {#1} \right]}
\newcommand{\EE}[2]{\mathop{\mathbf E}_{#1}\! \left[ {#2} \right]}

\newcommand{\negl}{{\sf negl}}

\newcommand{\DCRH}{dCRH\xspace}

\newcommand{\MyAtop}[2]{\genfrac{}{}{0pt}{}{#1}{#2}}
\newcommand{\N}{{\mathbb{N}}}

\newcommand{\R}{{\mathbb{R}}}

\newcommand{\bit}{\{0,1\}}

\newcommand{\aka} {also known as\ }

\newcommand{\ie}  {i.e.,\ }

\newcommand{\poly}{\mathsf{poly}}

\newcommand{\concat}{\circ}
\newcommand{\SD}{\mathbf{\Delta}}

\newcommand{\F}{\mathcal{F}}

\newcommand{\sender}{\mathcal{S}}
\newcommand{\receiver}{\mathcal{R}}
\newcommand{\verifier}{\mathcal{V}}
\newcommand{\decom}{\mathsf{decom}}
\newcommand{\com}{\mathsf{com}}

\newcommand{\view}{\mathsf{view}}
\newcommand{\ith}[1]{{#1}\textsuperscript{th}}

\newcommand{\set}[1]{\left\{#1\right\}}

\newcommand{\ignore}[1]{}

\newcommand{\class}[1]{\textsf{#1}\xspace}

\newcommand{\NP}{{\class{NP}}}

\newcommand{\supp}{\mathsf{supp}}
\newcommand{\ShanEnt}{\mathsf{H}}

\newcommand{\MaxEnt}{\mathsf{H}_{\mathsf{max}}}
\newcommand{\sA}{\mathsf{A}}

\newcommand{\KL}{\mathbf{D}_{\mathsf{KL}}}

\newcommand{\WI}{\mathcal{WI}}
\newcommand{\szkprob}{\Pi}
\newcommand{\SZK}{{\class{SZK}}}
\newcommand{\deci}{D}

\newcommand{\zo}{\{0,1\}}
\newcommand{\sbc}{\mathcal{SBC}}

\newcommand{\idc}{\mathcal{IDC}}

\newcommand{\hyb}{\mathcal{H}}

\newcommand{\pp}{c}
\newcommand{\Gc}{G}

\newcommand{\size}[1]{\left|#1\right|}
\renewcommand{\set}[1]{\{#1\}}
\renewcommand{\zo}{\set{0,1}}

\newcommand{\zs}{{\zo^\ast}}

\newcommand{\vy}{\mathbf{y}}
\newcommand{\vt}{\mathbf{t}}
\newcommand{\vY}{\mathbf{Y}}
\newcommand{\Gs}{{\widetilde{\Gc}}}

\newcommand{\tth}[1]{$#1$'th\xspace}
\renewcommand{\ith}{\tth{i}}

\def\getsr{\gets}

\newcommand{\Haccsam}{\operatorname{AccH}}
\newcommand{\Hrealsam}{\operatorname{RealH}}

\newcommand{\Supp}{\operatorname{Supp}}

\newcommand{\wrt} {with respect to\ }

\newcommand{\Col}{\operatorname{\sf Col}}

\title{Distributional Collision Resistance Beyond One-Way Functions}

\author{
Nir Bitansky
\thanks{School of Computer Science, Tel Aviv University. Email: \texttt{nirbitan@tau.ac.il}. Member of the Check Point Institute of Information Security. Supported by ISF grant 18/484, the Alon Young Faculty Fellowship, and by Len Blavatnik and the Blavatnik Family foundation.}
\and
Iftach Haitner
\thanks{School of Computer Science, Tel Aviv University. Email: \texttt{iftachh@cs.tau.ac.il}. Member of the Check Point Institute for Information Security. Research supported by ERC starting grant 638121.}
\and
Ilan Komargodski
\thanks{Cornell Tech, New York, NY. Email: \texttt{komargodski@cornell.edu}. Supported in part by an AFOSR grant FA9550-15-1-0262. }
\and{Eylon Yogev}
\thanks{Department of Computer Science, Technion. Email: \texttt{eylony@gmail.com}. Supported by the European Union's Horizon 2020 research and innovation program under grant agreement No.\ 742754.}
}

\date{}

\begin{document}
\maketitle

\begin{abstract}
    Distributional collision resistance is a relaxation of collision resistance that only requires that it is hard to sample a collision $(x,y)$ where $x$ is uniformly random and $y$ is uniformly random conditioned on colliding
  with $x$.
The notion lies between one-wayness and collision resistance, but its exact power is still not well-understood. On one hand, distributional collision resistant hash functions cannot be built from one-way functions in a black-box way, which may suggest that they are stronger. On the other hand, so far, they have not yielded any applications beyond one-way functions.

Assuming distributional collision resistant hash functions, we construct \emph{constant-round} statistically hiding commitment scheme. Such commitments are not known based on one-way functions and are impossible to obtain from one-way functions in a black-box way. Our construction relies on the reduction from inaccessible entropy generators to statistically hiding commitments by Haitner et al.\ (STOC '09). In the converse direction, we show that two-message statistically hiding commitments imply distributional collision resistance, thereby establishing a loose equivalence between the two notions.

  A corollary of the first result is that constant-round statistically hiding commitments are implied by average-case hardness in the class $\SZK$ (which is known to imply distributional collision resistance). This implication seems to be folklore, but to the best of our knowledge has not been proven explicitly. We provide yet another proof of this implication, which is arguably more direct than the one going through distributional collision resistance.
\end{abstract}

\thispagestyle{empty}
%\newpage
%\tableofcontents
\newpage
\setcounter{page}{1}

\section{Introduction}\label{sec:Introduction}

Distributional collision resistant hashing (\DCRH), introduced by Dubrov and
Ishai~\cite{DubrovI06}, is a relaxation of the notion of collision
resistance. In (plain) collision resistance, it is guaranteed that no efficient
adversary can find \emph{any} collision given a random hash function in the
family. In \DCRH, it is only guaranteed that no
efficient adversary can sample \emph{a random} collision given a random hash
function in the family. More precisely, given a random hash function
$h$ from the family, it is computationally hard to sample a
pair $(x,y)$ such that $x$ is uniform and $y$ is uniform in the preimage set
$h^{-1}(x)=\{z \colon h(x) = h(z)\}$. This hardness is captured by requiring that the adversary cannot get statistically-close to this distribution over collisions.\footnote{There are some subtleties in defining this precisely. The definition we use differs from previous ones~\cite{DubrovI06,HarnikN10,KomargodskiY18}.  We elaborate on the exact definition and the difference in the technical overview below and in \Cref{sec:dcrh}.}

\paragraph{The power of \DCRH.}
Intuitively, the notion of \DCRH seems quite weak. The adversary may even be able to sample collisions from the set of {\em all} collisions, but only from a skewed distribution, far
from the random one. Komargodski and Yogev \cite{KomargodskiY18} show that \DCRH can be constructed assuming average-case hardness in 
the complexity class {\em statistical zero-knowledge} (\SZK), whereas a similar implication is not known for multi-collision resistance.\footnote{Multi-collision resistance is another relaxation of collision resistance, where it is only hard to find multiple elements that all map to the same image. Multi-collision resistance does not imply \DCRH in a black-box way~\cite{KomargodskiNY18}, but Komargodski and Yogev~\cite{KomargodskiY18} give a non-black-box construction.} (let alone plain collision resistance). This can be seen as evidence suggesting that \DCRH may be weaker than collision resistance, or even multi-collision resistance \cite{KomargodskiNY17,BermanDRV18,BitanskyKP17,KomargodskiNY18}.

%\footnote{Formally, the results in \cite{KomargodskiY18} were stated for a slightly weaker notion of \DCRH, but they also hold for the definition we use in this paper.} \nir{for now removed the footnote although it's somewhat needed, we have to many...}

Furthermore, \DCRH has not led to the same cryptographic applications as collision resistance, or even multi-collision resistance. In fact, \DCRH has no known applications beyond those implied by one-way functions.

At the same time, \DCRH is not known to follow from one-way functions, and actually, cannot follow based on black-box reductions~\cite{Simon98}. In fact, it can even be separated from indistinguishability obfuscation (and one-way functions) \cite{AsharovS16}. Overall, we are left with a significant gap in our understanding of the power of \DCRH:

\begin{center}
  \textit{Does the power of \DCRH go beyond one-way functions?}
\end{center}

\subsection{Our Results}
We present the first application of \DCRH that is not known from one-way functions and is provably unachievable from one-way functions in a black-box way.

\begin{theorem}\label{thm:dcrh}
  \DCRH implies   \emph{constant-round} statistically hiding commitment scheme. 
\end{theorem}

Such commitment schemes cannot be constructed from one-way functions (or even permutations) in a black-box way due to a result of Haitner, Hoch, Reingold and Segev~\cite{HaitnerHRS15}. They show that the
number of rounds in such commitments must grow quasi-linearly in the security parameter.

The heart of Theorem \ref{thm:dcrh} is a construction of an inaccessible-entropy generator~\cite{HaitnerRVW09,HaitnerReVaWe18} from  \DCRH.

An implication of the above result is that constant-round statistically hiding
commitments can be constructed from
average-case hardness in $\SZK$. Indeed, it is known that such hardness implies the existence of a \DCRH~\cite{KomargodskiY18}.

\begin{corollary}\label{cor:szk}
  A Hard-on-average problem in $\SZK$ implies a
  \emph{constant-round} statistically hiding commitment scheme. 
\end{corollary}

The statement of Corollary~\ref{cor:szk} has been treated as known in several previous works (c.f.\ \cite{HaitnerRVW09,DvirGRV11,BitanskyDV17}), but a proof of this statement has so far not been published or (to the best of our knowledge) been publicly available. We also provide an alternative proof of this statement (and in particular, a different commitment scheme) that does not go through a construction of a \DCRH, and is arguably more direct.

% \ilan{Another implication of our main result (\Cref{thm:dcrh}) is that \DCRH  implies 
% a constant-round zero-knowledge proof for all \NP. 
% \begin{corollary}[who to cite? \iftach{GMW}]
% The existence of a \DCRH implies the existence of a constant-round zero-knowledge proof for all \NP.
% \end{corollary}
% }
% \eylon{I think the citation should be to \cite{BrassardCC88,GoldreichMW91} but it is only an argument and not a proof.}\ilan{Eylon, the point here is proofs. GMW gives O(1) round zk-args from owfs, right?}
% \nir{Guys, I don't think there's point in claiming as formal corollaries specific applications of constant round SCHCs. 
% In any case, the right citation is Goldreich-Kahan. GMW gives only constant soundness and is based on OWFs}\ilan{Nir, why not mention it? (I am ok with dropping the formal corollary)}

\paragraph{A limit on the power of \DCRH.}  We also show a converse connection between \DCRH and statistically hiding commitments. Specifically, we show that \emph{any} two-message statistically hiding commitment implies a \DCRH function family. 
\begin{theorem}
   Any two-message statistically hiding commitment scheme implies  \DCRH.
\end{theorem}

This establishes a loose equivalence between \DCRH and statistically hiding commitments. Indeed, the commitments we construct from \DCRH require more than two messages. Interestingly, we can even show that such commitments imply a stronger notion of \DCRH where the adversary's output distribution is not only noticeably far from the random collision distribution, but is $(1-\negl(n))$-far. 
%We note however  that the connection between \DCRH and constant-round statistically hiding we prove above is not tight, since the reduction of \cite{HaitnerReVaWe18} turns a two-block inaccessible entropy (and thus \DCRH) into 17-round statistically hiding commitment (and not two).

\subsection{Related Work on Statistically Hiding Commitments}
Commitment schemes, the digital analog of sealed envelopes, are central to cryptography. 
%
%Many of the most important
%cryptographic protocols rely on some form of commitments, perhaps most notably,
%secure multiparty computation~\cite{GennaroRR98} and the existence of a
%zero-knowledge protocol for membership in all languages in ${\sf
%  NP}$~\cite{GoldreichMW91,BrassardCC88}.\nir{Talk directly about the applications of statistically hiding commitments and not commitments in general; can do this later on, once you've noted the difference between the two.  Cite constant-round  ZK proofs [Goldreich-Kahan]. Can also add public-coin statitstical zero knowledge arguments and cite Barak01, and PassRosen05.}
%
More precisely, a commitment scheme is a two-stage interactive protocol between a sender $S$ and
a receiver $R$. After the commit stage, $S$ is bound to (at most)
one value, which stays hidden from $R$, and in the reveal stage $R$ learns this
value. The immediate question arising is what it means to be ``bound to'' and to
be ``hidden''. Each of these security properties can come in two main flavors,
either \emph{computational security}, where a polynomial-time adversary cannot
violate the property except with negligible probability, or the stronger notion
of \emph{statistical security}, where even an unbounded adversary cannot violate
the property except with negligible probability. However, it is known that there
do \emph{not} exist commitment schemes that are simultaneously statistically
hiding and statistically binding. 

There exists a one-message (i.e., non-interactive) statistically binding
commitment schemes assuming one-way permutations (Blum~\cite{Blum81}). From
one-way functions, such commitments can be achieved by a two-message protocol (Naor~\cite{Naor91} and H{\aa}stad, Impagliazzo, Levin and Luby~\cite{HILL99}).

Statistically hiding commitments schemes have proven to be somewhat more difficult to construct. Naor, Ostrovsky, Venkatesan and Yung~\cite{NaorOVY92} gave a
statistically hiding commitment scheme protocol based on one-way permutations,
whose linear number of rounds matched the lower bound of \cite{HaitnerHRS15} mentioned above. After
many years, this result was improved by Haitner, Nguyen, Ong, Reingold and Vadhan~\cite{HaitnerNORV09} constructing such commitment  based on the minimal hardness assumption that one-way functions exist. The reduction of \cite{HaitnerNORV09} was later simplified and made more efficient by Haitner, Reingold, Vadhan and Wee~\ \cite{HaitnerRVW09,HaitnerReVaWe18} to match, in some settings, the round complexity lower bound  of \cite{HaitnerHRS15}.  Constant-round statistically hiding commitment protocols are known to
exist based on families of collision resistant hash functions~\cite{NaorY89,DamgardPP93,HaleviM96}.  Recently, Berman, Degwekar, Rothblum and Vasudevan~\cite{BermanDRV18} and Komargodski, Naor and Yogev~\cite{KomargodskiNY18} constructed
constant-round statistically hiding commitment protocols assuming the existence of
 \emph{multi}-collision resistant hash functions.

Constant-round statistically hiding commitments are a basic building block in many fundamental applications. Two prominent examples are constructions of {\em constant-round} zero-knowledge proofs for all {\NP} (Goldreich and Kahan~\cite{GoldreichKahan}) and {\em constant-round} public-coin statistical zero-knowledge arguments for {\NP} (Barak~\cite{Barak01}, Pass and Rosen~\cite{PassR08}).

Statistically hiding commitment are also known to be tightly related to the hardness
of the class of problems that posses a statistical zero-knowledge protocol,
i.e., the class {\sf SZK}.  Ong and Vadhan~\cite{OngV08} showed that a language in
${\sf NP}$ has a zero-knowledge protocol if and only if the language has an
``instance-dependent'' commitment scheme. An instance-dependent commitment
scheme for a given language is a commitment scheme that can depend on an
instance of the language, and where the hiding and binding properties are
required to hold only on the YES and NO instances of the language, respectively.

\subsection{Directions for Future Work}

The security notions of variants of collision resistance, including plain collision resistance and multi-collision resistance, can be phrased in the language of entropy. For example, plain collision resistance requires that once a hash value $y$ is fixed the (max) entropy of preimages  that any efficient adversary can find is zero. In multi-collision resistance, it may be  larger than zero, even for every $y$, but still bounded by the size of allowed multi collisions. In distributional collision resistance, the (Shannon) entropy is close to maximal. 

Yet, the range of applications of collision resistance (or even multi-collision resistance) is significantly larger than those of distributional collision resistance. Perhaps the most basic such application is \emph{succinct} commitment protocols which are known from plain/multi-collision resistance but not from distributional collision resistance (by \emph{succinct} we mean that the total communication is shorter than the string being committed to). Thus, with the above entropy perspective in mind, a natural question is to characterize the full range or parameters between distributional and plain collision resistance and understand for each of them what are the applications  implied. A more concrete question is to find the minimal notion of security for collision resistance that implies succinct commitments.

    A different line of questions concerns understanding better the notion of distributional collision resistance and constructing it from more assumptions. Komargodski and Yogev constructed it from multi-collision resistance and from the average-case hardness of SZK. Can we construct it, for example, from the multivariate quadratic (MQ) assumption~\cite{MatsumotoI88}  or can we show an attack for random degree 2 mappings? Indeed, we know that random degree 2 mappings cannot be used for plain collision resistant hashing~\cite[Theorem 5.3]{ApplebaumHIKV17}.

%%% Local Variables:
%%% TeX-master: "main.tex"
%%% End:

\section{Technical Overview}\label{sec:Technical}

In this section, we give an overview of our techniques. We start with a more precise statement of the definition of \DCRH and a comparison with previous versions of its definition.

A \DCRH is a family of functions $\cH_n = \{h\colon \bit^n\to\bit^m \}$. (The functions are not necessarily compressing.) The security guarantee is that there exists a universal polynomial $p(\cdot)$ such that for every efficient adversary $\sA$  it holds that
\begin{align*}
    \SD\left((h,\sA(1^n,h)),(h,\Col(h))\right) \ge  \frac{1}{p(n)},
\end{align*}
where $\SD$ denotes statistical distance, $h\leftarrow \cH_n$ is chosen uniformly at random, and $\Col$ is a random variable that is sampled in the following way: Given $h$, first sample $x_1\leftarrow \bit^n$ uniformly at random and then sample $x_2$ uniformly at random from the set of all preimages of $x_1$ relative to $h$ (namely, from the set $\{x \colon h(x)=h(x_1)\}$). Note that $\Col$ may not be efficiently samplable and intuitively, the hardness of \DCRH says that there is no efficient way to sample from $\Col$, even approximately. 

Our definition is stronger than previous definitions of \DCRH~\cite{DubrovI06,HarnikN10,KomargodskiY18} by that we require the existence of a universal polynomial $p(\cdot)$, whereas previous definitions allow a different polynomial per adversary. Our modification seems necessary to get non-trivial applications of \DCRH, as the previous definitions are not known to imply one-way functions. In contrast, our notion of \DCRH implies distributional one-way functions which, in turn, imply one-way functions~\cite{ImpagliazzoL89} (indeed, the definition of distributional one-way functions requires a universal polynomial rather than one per adversary).\footnote{The previous definition is known to imply a weaker notion of distributional one-way functions (with a different polynomial bound per each adversary)~\cite{HarnikN10}, which is not known to imply one-way functions.} We note that previous constructions of \DCRH (from multi-collision resistance and $\SZK$-hardness)~\cite{KomargodskiY18} apply to our stronger notion as well.

\subsection{Commitments from \DCRH and Back}
We now describe our construction of constant-round statistically hiding commitments from \DCRH. To understand the difficulty, let us recall the standard approach to constructing statistically hiding commitments from (fully) collision resistant hash functions \cite{NaorY89,DamgardPP93,HaleviM96}. Here to commit to a bit $b$, we hash a random string $x$, and output $(h(x),s, b \oplus Ext_s(x))$, where $s$ is a seed for a strong randomness extractor $Ext$ and $b$ is padded with a (close to) random bit extracted from $x$. When $h$ is collision resistant, $x$ is computationally fixed and thus so is the bit $b$. However, for a \DCRH $h$, this is far from being the case: for any $y$, the sender might potentially be able to sample preimages from the set of all preimages.

The hash $h(x)$, however, does yield a weak binding guarantee. For simplicity of exposition, let us assume that any $y\in \zo^m$ has exactly $2^k$ preimages under $h$ in $\zo^n$. Then, for a noticeable fraction of commitments $y$, the adversary cannot open $y$ to a uniform $x$ in the preimage set $h^{-1}(y)$. In particular, the adversary must choose between two types of {\em entropy losses}: it either outputs a commitment $y$ of entropy $m'$ noticeably smaller than $m$, or after the commitment, it can only open to a value $x$ of entropy $k'$ noticeably smaller than $k$. One way or the other, in total $m'+k'$ must be noticeably smaller than $n=m+k$. This naturally leads us to the notion of {\em 
inaccessible entropy} defined by Haitner,  Reingold,  Vadhan and Wee~\cite{HaitnerRVW09,HaitnerReVaWe18}.

Let us briefly recall what inaccessible entropy is (see
\Cref{sec:prelim:inaccessibe} for a precise definition). The entropy of 
a random variable $X$ is a measure of ``the amount of randomness'' that $X$ 
contains. The notion of (in)accessible entropy measures the feasibility of sampling high-entropy strings that are {\em consistent} with a given random process. Consider the two-block generator (algorithm) $\Gc$ that samples $x \gets \zo^n$, and then outputs $y= h(x)$ and   $x$. The {\em real entropy} of $\Gc$ is defined as  the entropy of the generator's (total) output in a random execution, and is clearly equal to $n$, the length of $x$. The \emph{accessible entropy} of $\Gc$ measures the entropy of these output blocks from the point of view of an efficient \emph{$\Gc$-consistent} generator, which might act arbitrarily, but still outputs a value in the support of $\Gc$. 

Assume for instance that $h$ had been (fully) collision resistant.  Then from the point of view of any efficient $\Gc$-consistent generator $\Gs$, conditioned on its first block $y$, and its internal randomness,  its second output block is fixed (otherwise, $\Gc$ can be used for finding a collision). In other words, while the value of $x$ given $y$ may have entropy $k = n -m$, this entropy  is completely {\em inaccessible } for an efficient $\Gc$-consistent generator. (Note that we do not measure  here the entropy of the output blocks of $\Gs$, which clearly can be as high as the real entropy of $\Gc$ by taking $\Gs= \Gc$. Rather, we measure the entropy of the block from \emph{$\Gs$'s point of view}, and in particular, the entropy of its second block given the randomness used for generating the first block.). Haitner et al.\ show that any noticeable gap between the real entropy and the inaccessible entropy  of such an efficient generator can be leveraged for constructing statistically hiding commitments, with a number of rounds that is linear in the number of blocks.

Going back to \DCRH, we have already argued that in the simple case that $h$ is regular and onto $\zo^m$, we get a noticeable gap between the real entropy $n=m+k$ and the accessible entropy $m'+k' \leq m+k - 1/\poly(n)$. We prove that this is, in fact, true for any \DCRH: 

\begin{lemma}\label{thm:dcrhtpIAE}
  \DCRH implies a two-block inaccessible entropy generator.  
\end{lemma}
The block generator itself is  the simple generator described above: 
$$
\text{output $h(x)$ and then $x$, for  $x\gets \zo^n$}\enspace.
$$
%which is also the same generator that would be used for (fully) collision-resistant hash functions. 
The proof, however, is more involved than in the case of collision resistance. In particular, it is sensitive to the exact notion of entropy used. Collision resistant hash functions satisfy a very clean and simple guarantee --- the {\em maximum entropy}, capturing the support size, is always at most $m < n$. In contrast, for \DCRH (compressing or not), the maximum entropy could be as large as $n$, which goes back to the fact that the adversary may be able to sample from the set of {\em all} collisions (albeit from a skewed distribution). Still, we show a gap with respect to average (a.k.a Shannon) accessible entropy, which suffices for constructing statistically hiding commitments \cite{HaitnerReVaWe18}.

%\ilan{rewrote the following two paragraphs.} The construction of the  generator is the na\"ive construction we discussed above and already hinted why it should admit inaccessible entropy. That is, 
%the generator consists of two ``blocks'', $y=h(x)$ and $x$. Since $x$ is of length 
%$n$, the ``real'' entropy of the generator is $n$. We show that any efficient adversary can at most generate noticeably less bits of entropy.

%In the case of plain collision resistance, it is trivial to show this as the second message of the generator has entropy 0, leading to an entropy gap. However, in the case of \DCRH, the marginal distribution of each of the two messages could be maximal. Indeed, an adversary might choose a $y$ for which it knows all preimages and then sample one uniformly at random, or it can sample $y$ at random and then output some preimage (not necessarily a uniform one in the preimage set). So, we have to take into account the \emph{joint distribution} of both messages. We do so in \Cref{our-generator}, where our proof uses various tools from information theory.

\paragraph{From commitments back to \DCRH.}
We show that any two-message statistically hiding commitment implies a \DCRH function family. %For concreteness, let us focus here on the case that the commitment scheme supports single bits. 
Let $(\sender,\receiver)$ be the sender and receiver of a statistically hiding bit commitment. The first message sent by the receiver is the description of the hash function: $h\leftarrow \receiver(1^n)$. The sender's commitment to a bit $b$, using randomness $r$, is the hash  of $x=(b,r)$. That is, $h(x) = \sender(h,b;r)$. 

To argue that this is a \DCRH, we show that any attacker that can sample collisions that are close to the random collision distribution $\Col$ can also break the binding of the commitment scheme. For this, it suffices to show that a collision $(b,r),(b',r')$ sampled from $\Col$, translates to equivocation --- the corresponding commitment can be opened to two distinct bits $b\neq b'$. Roughly speaking, this is because statistical hiding implies that a random collision to a random bit $b$ (corresponding to a random hash value) is statistically independent of the underlying committed bit. In particular, a random preimage of such a commitment will consist of a different bit $b'$ with probability roughly $1/2$. See details in \Cref{sec:dcrh_from_stat_hiding}.

%For that, we need to show that if the commitment scheme is secure, then no efficient adversary can come up with a random collision. We do so by way of contradiction, assuming that such an adversary exists and using it to equivocate the commitment scheme. Assume for simplicity that the \DCRH adversary is the (inefficient) optimal algorithm $\Col$ and that the commitment scheme is perfectly hiding. We run $\Col$ on $h$ to get a collision $x_1=(b,r)$ and $x_2=(b',r')$. The fact that both can be used as legal openings is immediate by definition of the hash function, so it remains to show that $b\neq b'$ (which implies that indeed these are two openings for different messages). Recall that $(b,r)$ is chosen uniformly at random and then $(b',r')$ is chosen uniformly at random conditioned on colliding with $(b,r)$. By perfect hiding of the commitment scheme, we can sample $(b',r')$ uniformly at random conditioned on colliding with $(0,r)$. Thus, $b$ and $b'$ are independent and $b$ is uniformly random so the probability that they are the same is roughly $1/2$ which is enough to break the binding of the commitment scheme.

%The full details of the proof where we also get rid of the simplifying assumption mentioned above is given in \Cref{sec:dcrh_from_stat_hiding}. We also remark on what happens if we start with a statistically hiding commitment that works for strings rather than bits.

\subsection{Commitments from SZK Hardness}

We now give an overview of our construction of statistically hiding commitments directly from average-case hardness in $\SZK$. Our starting point is a result of Ong and Vadhan~\cite{OngV08} showing that any promise problem in $\SZK$ has an {\em instance-dependent commitment.} These are commitments that are also parameterized by an instance $x$, such that if $x$ is a {\em yes instance}, they are statistically hiding and if $x$ is a {\em no instance}, they are statistically binding. We construct statistically hiding commitments from instance-dependent commitments for a hard-on-average problem $\Pi=(\Pi_N,\Pi_Y)$ in $\SZK$. 

\paragraph{A first attempt: using zero-knowledge proofs.} To convey the basic idea behind the construction, let us first assume that $\Pi$ satisfies a strong form of average-case hardness where we can efficiently sample no-instances from $\Pi_N$ and yes-instances from $\Pi_Y$ so that the two distributions are computationally indistinguishable. Then a natural protocol for committing to a message $m$ is the following: The receiver $\receiver$ would sample  a yes-instance $x\gets \Pi_Y$, and send it to the sender $\sender$ along with zero-knowledge proof \cite{GMR} that $x$ is indeed a yes-instance. The sender $\sender$ would then commit to $m$ using an $x$-dependent commitment.

To see that the scheme is statistically hiding, we rely on the soundness of the proof which guarantees that $x$ is indeed a yes-instance, and then on the hiding of the instance-dependent scheme. To prove (computational) binding, we  rely on zero knowledge property and the hardness of $\Pi$. Specifically, by zero knowledge, instead of sampling $x$ from $\Pi_Y$, we can sample it from any computationally indistinguishable distribution, without changing the probability that an efficient malicious sender breaks binding. In particular, by the assumed hardness of $\Pi$, we can sample $x$ from $\Pi_N$. Now, however, the instance-dependent commitment guarantees binding, implying that the malicious sender will not be able to equivocate.

The main problem with this construction is that constant-round zero-knowledge proofs (with a negligible soundness error) are only known assuming constant-round statistically hiding commitments \cite{GoldreichKahan}, which is exactly what we are trying to construct.

\paragraph{A second attempt: using witness-indistinguishable proofs.} Instead of relying on zero-knowledge proofs, we rely on the weaker notion of witness-indistinguishable proofs and use the {\em independent-witnesses paradigm} of Feige and Shamir~\cite{FeigeShamir90}. (Indeed such proofs are known for all of $\NP$, based average-case hardness in \SZK~\cite{GMW87,Naor91,OstrovskyW93}, see Section \ref{sec:szk} for details.) We change the previous scheme as follows: the receiver $\receiver$ will now sample {\em two} instances $x_0$ and $x_1$ and provide a witness-indistinguishable proof that at least one of them is a yes-instance. The sender, will secret share the message $m$ into two random messages $m_0,m_1$ such that $m=m_0\oplus m_1$, and return two instance-dependent commitments to $m_0$ and $m_1$ relative to $x_0$ and $x_1$, respectively.

Statistical hiding follows quite similarly to the previous protocol --- by the soundness of the proof one of the instances $x_b$ is a yes-instance, and by the hiding of the $x_b$-dependent commitment, the corresponding share $m_b$ is statistically hidden, and thus so is $m$. To prove binding, we first note that by witness indistinguishability, to prove its statement, the receiver could use $x_b$ for either $b\in\zo$. Then, relying on the hardness of $\Pi$, we can sample $x_{1-b}$ to be a no-instance instead of a yes-instance. If $b$ is chosen at random, the sender cannot predict $b$ better than guessing. At the same time, in order to break binding, the sender must equivocate with respect to at least one of the instance-dependent commitments, and since it cannot equivocate with respect to the no-instance $x_{1-b}$, it cannot break binding unless it can get an advantage in predicting $b$.

\paragraph{Our actual scheme.} The only gap remaining between the scheme just described and our actual scheme is our assumption regarding the strong form of average-case hardness of $\Pi$. In contrast, the standard form of average-case hardness only implies a single samplable distribution $D$, such that given a sample $x$ from $D$ it is hard to tell whether $x$ is a yes-instance or a no-instance better than guessing. 

This requires the following  changes to the protocol. First, lacking a samplable distribution on yes-instances, we consider instead the product distribution $D^n$, as a way to sample {\em weak yes instances} --- $n$-tuples of instances where at least one is a yes-instance in $\Pi_Y$. Unlike before, where everything in the support of the yes-instance sampler was guaranteed to be a yes-instance, now we are only guaranteed that a random tuple is a weak yes instance with overwhelming probability. To deal with this weak guarantee, we add a {\em coin-tossing into the well} phase \cite{GMW87}, where the randomness for sampling an instance from $D^n$ is chosen together by the receiver and sender. We refer the reader to Section \ref{sec:szk} for more details.

\section{Preliminaries}\label{sec:prelims}
Unless stated otherwise, the logarithms in this paper are base 2. For a
distribution $\cD$ we denote by $x \leftarrow \cD$ an element chosen from
$\cD$ uniformly at random. For an integer $n \in \mathbb{N}$ we denote by
$[n]$ the set $\{1,\ldots, n\}$. We denote by $U_n$ the uniform distribution
over $n$-bit strings.
We denote by $\concat$ the string concatenation operation. A function $\negl\colon\N\to\R^+$ is \emph{negligible} if for every constant
$c > 0$, there exists an integer $N_c$ such that $\negl(n) < n^{-c}$ for all
$n > N_c$.
\subsection{Cryptographic Primitives}\label{sec:prelim:uowhf}
A function $f$, with input length $m_1(n)$ and outputs length $m_2(n)$,
specifies for every $n\in \N$ a function
$f_n\colon\bit^{m_1(n)}\to\bit^{m_2(n)}$. We only consider functions with
polynomial input lengths (in $n$) and occasionally abuse notation and write
$f(x)$ rather than $f_n(x)$ for simplicity. The function $f$ is computable in
polynomial time (efficiently computable) if there exists a probabilistic machine that
for any $x \in \bit^ {m_1(n)}$ outputs $f_n(x)$ and runs in time polynomial in
$n$.

A function family ensemble is an infinite set of function families, whose
elements (families) are indexed by the set of integers. Let $\F = \{\F_n\colon
\mathcal D_n\to\mathcal R_n\}_{n\in \N}$ stand for an ensemble of function
families, where each $f\in \F_n$ has domain $\mathcal D_n$ and range $\mathcal
R_n$. An efficient function family ensemble is one that has an efficient
sampling and evaluation algorithms.
\begin{definition}[Efficient function family ensemble]
  A function family ensemble $\F = \{\F_n\colon
\mathcal D_n\to\mathcal R_n\}_{n\in \N}$ is efficient if:
\begin{itemize}
\item $\F$ is samplable in polynomial time: there exists a probabilistic
  polynomial-time machine that given $1^n$, outputs (the description of) a
  uniform element in $\F_n$.
\item There exists a deterministic algorithm that given $x \in \mathcal D_n$ and
  (a description of) $f \in \mathcal F_n$, runs in time $\poly(n, |x|)$ and
  outputs $f(x)$.
\end{itemize}
\end{definition}

\subsection{Distance and Entropy Measures}
\begin{definition}[Statistical distance]
  The \emph{statistical distance} between two random variables $X,Y$ over a
  finite domain $\Omega$, is defined by
  \begin{align*}
    \SD(X,Y) \triangleq \frac{1}{2}\cdot \sum_{x\in \Omega}^{} \left|{\pr{X=x} -
    \pr{Y=x}}\right|.
  \end{align*}
  We say that $X$ and $Y$ are $\delta$-close (resp.\ -far) if
  $\SD(X,Y)\leq \delta$ (resp.\ $\SD(X,Y)\geq \delta$).
\end{definition}

\paragraph{Entropy.} Let $X$ be a random variable. For any $x\in \supp(X)$, the sample-entropy of $x$ with respect to $X$ is
\begin{align*}
 \ShanEnt_X(x) = \log \left(\frac{1}{\pr{X=x}}\right).
\end{align*} 
The Shannon entropy of $X$ is defined as:
\begin{align*}
 \ShanEnt(X) = \EE{x\leftarrow X}{\ShanEnt_X(x)}.
\end{align*}
%
%\item The min-entropy of $X$ is defined as:
%\begin{align*}
% \MinEnt(X) = \min_{x\in \supp(X)}{\ShanEnt_X(x)}.
%\end{align*}
% 
%\item The max-entropy of $X$ is defined as:
%\begin{align*}
% \MaxEnt(X) = \log |\supp(X)|.
%\end{align*}
%  
%\end{itemize}

%The following fact is immediate from the definitions.
%\begin{fact} 
%  It holds that $\MinEnt(X) \le \ShanEnt(X) \le \MaxEnt(X)$. Moreover, %equality holds if
% and only if $X$ is uniform over a subset of its domain.
%\end{fact}

\paragraph{Conditional entropy.} Let $(X,Y)$ be a jointly distributed random
variable.
\begin{itemize}
\item  For any $(x,y)\in \supp(X,Y)$, the conditional sample-entropy to be
\begin{align*}
 \ShanEnt_{X\mid Y}(x \mid y) = \log \left(\frac{1}{\pr{X=x \mid Y=y}}\right).
\end{align*}

\item The conditional Shannon entropy is 
\begin{align*}
 \ShanEnt(X\mid Y) = \EE{(x,y)\leftarrow (X,Y)}{\ShanEnt_{X\mid Y}(x \mid y)} =
 \EE{y\leftarrow Y}{\ShanEnt(X|_{Y=y})} = \ShanEnt(X,Y) - \ShanEnt(Y).
\end{align*}
 
\end{itemize}

\paragraph{Relative entropy.} We also use basic facts about relative entropy
(\aka, Kullback-Leibler divergence). 

 \begin{definition}[Relative entropy]
   Let $X$ and $Y$ be two random variables over a finite domain $\Omega$.  The \emph{relative entropy}   is 
   \begin{align*}
     \KL(X \| Y) = \sum_{x\in\Omega}\pr{X=x}\cdot \log\left( 
     \frac{\pr{X=x}}{\pr{Y=x}}\right).
   \end{align*}
 \end{definition}

 \begin{proposition}[Chain rule]\label{prop:kl_chain_rule}
   Let $(X_1,X_2)$ and $(Y_1,Y_2)$ be random variables. It holds that 
   \begin{align*}
     \KL((X_1,X_2)\|(Y_1,Y_2)) =  \KL(X_1 \| Y_1) + \EE{x\leftarrow X_1}{\KL(X_2|_{X_1=x} \| Y_2|_{Y_1=x})}.
   \end{align*}
 \end{proposition}

A well-known relation between statistical distance and relative entropy is
 given by Pinsker's inequality.
 \begin{proposition}[Pinsker's inequality]\label{prop:pinsker} 
   For any two random variables $X$ and $Y$ over a finite domain it holds
   that
   \begin{align*}
     \SD(X,Y) \le \sqrt{\frac{\ln{2}}{2}\cdot \KL(X \| 
     Y)}.  
   \end{align*}
 \end{proposition}

 Another useful inequality is Jensen's inequality.
 \begin{proposition}[Jensen's inequality]\label{prop:jensen}
   If $X$ is a random variable and $f$ is concave, then
   \begin{align*}
     \E{f(X)} \leq f(\E{X}).
   \end{align*}
 \end{proposition}

% We will use another (less well-known) distance measure called the
% \emph{triangular discrimination} (a.k.a Le Cam Divergence).
% \begin{definition}[Triangular discrimination]\label{def:td}
%   The \textsf{triangular discrimination} between two random variables $X,Y$ over
%   a finite domain $\Omega$, is defined by
%   \begin{align*}
%     \TD(X,Y) = \sum_{x\in \Omega} \frac{(\Pr[X=x] - \Pr[Y=x])^2}{\Pr[X=x] + \Pr[Y=x]}
%   \end{align*}
% \end{definition}
% It is known that the triangular discrimination is bounded from above by
% the statistical distance and from below by the statistical
% distance squared (see, for example, \cite[Eq.~(2.11)]{Topsoe00}). 
% \begin{proposition} \label{lemma:sd_td}
% For any two random variables $X,Y$
%   over the same finite domain, it holds that
%   \begin{align*}
%     2\cdot\SD(X, Y)^2 \leq  \TD(X,Y) \leq 2\cdot\SD(X, Y).
%   \end{align*}
% \end{proposition}

\subsection{Commitment Schemes}\label{sec:prelim:commit}

A commitment scheme is a two-stage interactive protocol between a sender
$\sender$ and a receiver $\receiver$. The goal of such a scheme is that after
the first stage of the protocol, called the commit protocol, the sender is bound
to at most one value. In the second stage, called the opening protocol, the
sender opens its committed value to the receiver. Here, we are interested in statistically hiding and computationally binding commitments. Also, for simplicity, we restrict our attention to protocols that can be used to commit to bits (i.e., strings of length 1). 

In more detail, a commitment scheme   is defined
via a pair of probabilistic polynomial-time algorithms $(\sender,
\receiver, \verifier)$ such that:
\begin{itemize}
\item The commit protocol: $\sender$ receives as input the security parameter
  $1^n$ and a bit $b\in \bit$. $\receiver$ receives as input the
  security parameter $1^n$. At the end of this stage, $\sender$ outputs
  $\decom$ (the  decommitment) and $\receiver$
  outputs $\com$ (the commitment).
\item The verification: $\verifier$ receives as input the security
  parameter $1^n$, a commitment $\com$, a
  decommitment $\decom$, and outputs either a bit $b$ or $\bot$.
\end{itemize}

A commitment scheme is {\em public coin} if all messages sent by the receiver are independent random coins. 

Denote by $(\decom,\com)\leftarrow \langle \sender (1^n, b),
\receiver \rangle$ the experiment in which $\sender$ and $\receiver$ interact
with the given inputs and uniformly random coins, and eventually $\sender$
outputs a decommitment string and $\receiver$ outputs a
commitment. The completeness of the protocol says that for all $n\in \N$, every
$b\in\bit$, and every tuple $(\decom,\com)$ in
the support of $\langle \sender (1^n, b), \receiver \rangle$, it holds that $\verifier (\decom,\com) = b$. Unless otherwise stated, $\verifier$ is the canonical verifier that receives the sender's coins as part of the decommitment and checks their consistency with the transcript.

Below we define two security properties one can require from a commitment
scheme. The properties we list are \emph{statistical-hiding} and
\emph{computational-binding}. These roughly say that after the commit stage, the
sender is \emph{bound} to a specific value but the receiver cannot know this
value.

\begin{definition}[binding]
  A commitment scheme $(\sender, \receiver, \verifier)$ is
  binding if for every probabilistic polynomial-time adversary $\sender^*$
  there exits a negligible function $\negl(n)$ such that
  \begin{align*}
    \pr{\MyAtop{ \verifier(\decom,\com)=0
        \text{ and }}{\verifier(\decom',\com)=1 } \;:\; (\decom,\decom',\com) \leftarrow \langle \sender^* (1^n),
        \receiver\rangle}\leq \negl(n)
  \end{align*}
  for all  $n\in \N$, where the probability is taken over the random
  coins of both $\sender^*$ and $\receiver$.
\end{definition}

Given a commitment scheme $(\sender, \receiver,
\verifier)$ and an adversary $\receiver^*$, we denote by $\view_{\langle
  \sender(b), \receiver^*\rangle}(n)$ the distribution on the view of
$\receiver^*$ when interacting with $\sender(1^n, b)$. The view consists of
$\receiver^*$'s random coins and the sequence of messages it received from
$\sender$. The distribution is taken over the random coins of both $\sender$ and
$\receiver$. Without loss of generality, whenever $\receiver^*$ has no
computational restrictions, we can assume it is deterministic.
\begin{definition}[hiding]
  A commitment scheme $(\sender, \receiver, \verifier)$ is statistically hiding if there exists a negligible function $\negl(n)$ such that for every (deterministic) adversary $\receiver^*$
  %and every distinct $s_0,s_1\in \bit^\ell$,
  it holds that
  \begin{align*}
    \SD\left(\{\view_{\langle \sender(0),
    \receiver^*\rangle}(n)\}, \{\view_{\langle \sender(1),
    \receiver^*\rangle}(n)\}\right) \leq \negl(n)
  \end{align*}
  for all  $n\in \N$.
\end{definition}

%\paragraph{Other properties.}\nir{I suggest to remove this remark, it's %rrelvant to the paper} In this work we are mostly interested in the round %complexity of commitment protocol. More specifically, we are interested in %protocols that work in a constant number of rounds. Note that there are %additional features of interest for commitment schemes. One such feature is %\emph{succinctness} which says that the commitment scheme allows to commit %on a very long string yet the commitment itself is (much) shorter than that. %Another feature is called \emph{local opening} and it says that one can %verify/open a small part of the committed string without revealing the whole %input. The applications of the latter two features are beyond the scope of %this work, but see~\cite{KomargodskiNY18,BitanskyKP17} for some examples.

\subsection{Distributional Collision Resistant Hash Functions}\label{sec:dcrh}
Roughly speaking, a distributional collision resistant hash function \cite{DubrovI06} guarantees that no efficient adversary can sample a uniformly random
collision. We start by defining more precisely what we mean by a random collision throughout the paper, and then move to the actual definition.

%This relaxation of classical collision resistance was introduced by Dubrov and Ishai~.

\begin{definition}[Ideal collision finder]
 Let $\Col$ be the random function that given a (description) of a function $h\colon\bit^{n}\to\bit^m$ as input, returns a collision  $(x_1,x_2)$ \wrt $h$ as follows: it samples a uniformly random element, $x_1 \leftarrow 
\bit^{n}$,  and then samples a uniformly random element that collides with $x_1$ under $h$, $x_2 \leftarrow \{x \in
\bit^{n} \colon h(x)=h(x_1)\}$. (Note that possibly, $x_1=x_2$.)
\end{definition}

\begin{definition}[Distributional collision resistant hashing]\label{def:dcrh}
  Let $\cH = \{ \cH_n\colon\allowbreak \bit^{n} \to \bit^{m(n)} \}_{n\in \N}$ be
  an efficient function family ensemble.  We say that $\cH$ is a
   secure \emph{distributional collision resistant hash} ($\DCRH$) function
  family if there exists a polynomial $p(\cdot)$ such that for any probabilistic polynomial-time algorithm $\sA$, it holds that
  \begin{align*}
    \SD\left((h,\sA(1^n,h)),(h,\Col(h))\right) \ge  \frac{1}{p(n)},
  \end{align*}
  for $h\leftarrow \cH_n$ and large enough $n\in\N$.
  \end{definition}

\paragraph{Comparison with the previous definition.} Our definition deviates from the previous definition of distributional collision resistance considered in~\cite{DubrovI06,HarnikN10,KomargodskiY18}. The definition in the above-mentioned works is equivalent to requiring that for any efficient adversary $\sA$, there exists a polynomial $p_{\sA}$, such that the collision output by $\sA$ is $\frac{1}{p_{\sA}(n)}$-far from a random collision on average (over $h$). Our definition switches the order of quantifiers, requiring that there is one such polynomial $p(\cdot)$ for all adversaries $\sA$.

We note that the previous definition is, in fact, not even known to imply one-way functions. In contrast, the definition presented here strengthens that of {\em distributional one-way functions}, which in turn implies one-way functions~\cite{ImpagliazzoL89}. 
Additionally, note that both constructions of distributional collision resistance in~\cite{KomargodskiY18} (from multi-collision resistance and from \SZK~hardness) satisfy our stronger notion of security (with a similar proof).

\paragraph{On compression.} As opposed to classical notions of collision resistance (such as plain collision resistance or multi-collision resistance), it makes sense to require distributional collision resistance even for \emph{non-compressing} functions. So we do not put a restriction on the order between $n$ and $m(n)$. As a matter of fact, by padding, the input, arbitrary polynomial compression can be assumed without loss of generality.

%%% Local Variables:
%%% TeX-master: "main.tex"
%%% End:

\section{From \DCRH to Statistically Hiding Commitments and Back}

We show distributional collision resistant hash functions imply constant-round statistically hiding commitments. 
\begin{theorem}\label{thm:stat_hide_dcrh}
  Assume the existence of a distributional collision resistant hash function family. Then, there exists a constant-round statistically hiding and
  computationally binding commitment scheme. 
\end{theorem}

Our proof relies on the transformation of Haitner et al.~\cite{HaitnerRVW09,HaitnerReVaWe18}, translating inaccessible-entropy generators to statistically hiding commitments. Concretely, we construct appropriate inaccessible-entropy generators from distributional collision resistant hash functions.
In \Cref{sec:prelim:inaccessibe}, we recall the necessary definitions and the result of \cite{HaitnerReVaWe18}, and then in \Cref{our-generator}, we prove Theorem~\ref{thm:stat_hide_dcrh}. 

We complement the above result by showing a loose converse to \Cref{thm:stat_hide_dcrh}, namely that two message statistically 
hiding commitments (with possibly large communication) imply the existence of 
distributional collision resistance hashing.
\begin{theorem}\label{thm:dcrh_from_stat_hiding}
Assume the existence of a binding and statistically hiding  two-message commitment scheme. Then, there exists a \DCRH function family.
\end{theorem}
This proof of \Cref{thm:dcrh_from_stat_hiding} appears in \Cref{sec:dcrh_from_stat_hiding}.

\subsection{Preliminaries on Inaccessible Entropy Generators}\label{sec:prelim:inaccessibe}
The following  definitions of real and accessible entropy of protocols are taken from \cite{HaitnerReVaWe18}.
\begin{definition}[Block generators]\label{def:blockGenrator}
	Let $n$ be a security parameter, and let $\pp=\pp(n)$, $s=s(n)$ and $m=m(n)$. An {\sf $m$-block generator} is a function $\Gc \colon \zo^{\pp} \times \zo^{s} \mapsto (\zs)^{m}$. It is {\sf efficient} if its running time on input of length $\pp(n)+ s(n)$ is polynomial in $n$.
	
	%let $\ell(n) = \sum_{i\in [m(n)]} \ell_i(n)$, and
	We   call parameter $n$ the {\sf security parameter}, $\pp$ the {\sf public parameter length}, $s$ the {\sf seed length},  $m$  the  {\sf number of blocks}, and  $\ell(n) = \max_{(z,x)\in \zo^{\pp(n)} \times \zo^{s(n)},i\in [m(n)]} \size{\Gc(z,x)_i}$ the {\sf maximal block length} of $\Gc$.  
\end{definition}

\begin{definition}[Real sample-entropy]\label{def:RealSamEnt}
Let $\Gc$ be an $m$-block generator over $\zo^{\pp} \times \zo^{s}$,  let $n\in\N$, let  $Z_n$ and $X_n$ be uniformly distributed over $\zo^{\pp(n)}$ and $\zo^{s(n)}$, respectively, and let $\vY_n= (Y_{1},\ldots,Y_m) =  \Gc(Z_n, X_n)$.  For $n\in\N$ and $i\in [m(n)]$, define the {\sf real sample-entropy of $\vy\in \Supp(Y_{1},\ldots,Y_i)$ given $z\in\Supp(Z_n)$} as
	$$\Hrealsam_{\Gc,n}(\vy| z) = \sum_{j=1}^i \ShanEnt_{Y_j|Z_n,Y_{<j}}(\vy_j|z,\vy_{<j}).$$
\end{definition}
\noindent We omit the security parameter from the above notation when clear from the context.

\begin{definition}[Real entropy]
	Let $\Gc$ be an $m$-block generator, and  let  $Z_n$ and $\vY_n$ be as in Definition~\ref{def:RealSamEnt}. Generator $\Gc$ has {\sf real entropy at least $k = k(n)$}, if
	$$\EE{(z,\vy) \getsr (Z_n,\vY_n)}{\Hrealsam_{\Gc,n}(\vy| z)} \geq k(n)$$
	for every $n\in \N$.
	
	The generator $\Gc$ has {\sf real min-entropy at least $k(n)$ in its \ith block} for some $i = i(n)\in [m(n)]$, if
	$$\prob{(z,\vy) \getsr (Z_n,\vY_n)}{\ShanEnt_{Y_i|Z_n,Y_{<i}}(\vy_i|z,\vy_{<i}) < k(n)} = \negl(n).$$

We say the above bounds	 are {\sf invariant to the public parameter} if they hold for any fixing of the public parameter $Z_n$.\footnote{In particular, this is the case when there is no public parameter, \ie $\pp= 0$.}
\end{definition}

It is known that the real Shannon entropy amounts to measuring the
standard conditional Shannon entropy of $\Gc$'s output blocks.
\begin{lemma}[{\cite[Lemma 3.4]{HaitnerReVaWe18}}]
	Let  $\Gc$,  $Z_n$ and $\vY_n$ be as in \cref{def:RealSamEnt} for some $n\in \N$, then
	$$\EE{(z,\vy) \getsr (Z_n,\vY_n)}{\Hrealsam_{\Gc,n}(\vy|z)} =  \ShanEnt(\vY_n|Z_n).$$
\end{lemma}

Toward the definition of \emph{inaccessible entropy}, we first define \emph{online block-generators} which are a special type of block generators that toss fresh random coins before
outputting each new block.

\begin{definition}[Online block generator]
	Let $n$ be a security parameter, and let $\pp = \pp(n)$ and $m=m(n)$. An $m$-block  {\sf online} generator is a function $\Gs \colon \zo^\pp \times (\zo^{v})^{m} \mapsto  (\zs)^{m}$ for some $v =v(n)$, such that the \ith output block of $\Gs$ is a function of (only) its first $i$  input blocks. We denote the {\sf transcript} of $\Gs$ over random input by  $T_{\Gs}(1^n) = (Z,R_1,Y_1,\ldots,R_m,Y_m)$, for  $Z \getsr \zo^\pp$, $(R_1,\ldots,R_m) \getsr (\zo^{v})^{m}$ and $(Y_1,\ldots,Y_m) =\Gs(Z,R_1,\ldots,R_i)$.
\end{definition}

That is, an online block generator is a special type of block generator that tosses fresh random coins before outputting each  new block. In the following, we let $\Gs(z,r_1,\ldots,r_i)_i$  stand for $\Gs(z,r_1,\ldots,r_i,x^\ast)_i$ for arbitrary $x^\ast \in (\zo^{v})^{m -i}$ (note that the choice of $x^\ast$ has no effect on the value of $\Gs(z,r_1,\ldots,r_i,x^\ast)_i$).

\begin{definition}[Accessible sample-entropy]\label{def:AccessibleSampleEntropy}
	Let $n$ be a security parameter, and let $\Gs$ be an online $m=m(n)$-block online  generator. The {\sf accessible sample-entropy of $\vt =(z,r_1,y_1,\ldots,r_m,y_m)\in \Supp(Z,R_1,Y_1\ldots,R_m,Y_m) = T_{\Gs}(1^n)$} is defined by
	
	$$\Haccsam_{\Gs,n}(\vt) = \sum_{i=1}^{m} \ShanEnt_{Y_i|Z,R_{<i}}(y_i|z,r_{<i}).$$
\end{definition}
\noindent 
Again, we omit the security parameter from the above notation when clear from the context.

As in the case of real entropy, the expected accessible entropy of a random transcript can  be expressed in terms of the standard conditional Shannon entropy.
\begin{lemma}[{\cite[Lemma 3.7]{HaitnerReVaWe18}}]\label{lemma:acc_entr}
	Let $\Gs$ be an online $m$-block generator  and let  $(Z,R_1,Y_1,\ldots,\allowbreak R_m,\allowbreak Y_m) = T_\Gs(1^n)$ be its transcript.	Then, 
	$$\EE{\vt\getsr T_{\Gs}(Z,1^n)} {\Haccsam_{\Gs}(\vt)} = \sum_{i\in[m]} \ShanEnt(Y_i|Z,R_{<i}).$$
\end{lemma}

We focus on efficient generators that are consistent with respect to $\Gc$.
That is, the support of their output is contained in that of $\Gc$.

\begin{definition}[Consistent  generators]\label{def:nonFailingGen}
	Let $\Gc$ be a block generator over $\zo^{\pp(n)} \times \zo^{s(n)}$. A block (possibly online) generator $\Gc'$ over $\zo^{\pp(n)} \times \zo^{s'(n)}$ is {\sf $\Gc$ consistent} if, for every $n\in \N$, it holds that  $\Supp(\Gc'(U_{\pp(n)},U_{s'(n)})) \subseteq \Supp(\Gc(U_{\pp(n)},U_{s(n)}))$.
\end{definition}

\begin{definition}[Accessible entropy]\label{def:accessible-entropy}
	A block generator $\Gc$ has {\sf accessible entropy at most $k = k(n)$} if, for  every efficient $\Gc$-consistent, online generator $\Gs$ and all large enough $n$,
	$$\EE{\vt\getsr T_\Gs(1^n)} {\Haccsam_\Gs(\vt)} \leq k.$$
\end{definition}

We call a  generator whose real entropy is noticeably higher than it  accessible entropy an inaccessible entropy generator.  

We use the following reduction from inaccessible entropy generators to constant round statistically   hiding commitment.

\begin{theorem}[{\cite[Thm.\ 6.24]{HaitnerReVaWe18}}]
	Let  $\Gc$ be an efficient  block generator  with constant number of blocks.  Assume   $\Gc$'s real Shannon entropy 
	is at least $k(n)$ for some efficiently computable function $k$, and that its accessible entropy is bounded by $k(n) - 1/p(n)$ for some  $p\in \poly$. Then  there exists a constant-round  statistically hiding and computationally binding commitment scheme. Furthermore, if the bound on the real entropy is invariant to the public parameter, then  the commitment is receiver public-coin.
\end{theorem}

\begin{remark}[Inaccessible max/average entropy]
Our result relies on the reduction from inaccessible \emph{Shannon} entropy generators to statistically hiding commitments, given  in \cite{HaitnerReVaWe18}. The proof of this reduction follows closely the proof in previous versions~\cite{HaitnerVadhan17,HaitnerRVW09}, where the reduction was from inaccessible \emph{max} entropy generators.
The extension to Shannon entropy generators is essential for our result.
\end{remark}

\subsection{From \DCRH to Inaccessible Entropy Generators -- Proof of Theorem~\ref{thm:stat_hide_dcrh}}\label{our-generator}
In this section we show that there is a block generator with two blocks in which there
is a gap between the real entropy and the accessible entropy.  Let $\cH = \{
\cH_n\colon\allowbreak \bit^{n} \to \bit^{m} \}_{n\in \N}$ be a \DCRH for $m = m(n)$ and assume that each $h\in \cH_n$ requires $c=c(n)$ bits to describe.  By Definition~\ref{def:dcrh}, there exists a polynomial $p(\cdot)$ such that for any probabilistic polynomial-time algorithm $\sA$, it holds that
  \begin{align*}
        \SD\left((h,\sA(1^n,h)),(h,\Col(h))\right) = 
        \EE{h\leftarrow \cH_n}{\SD\left(\sA(1^n,h),\Col(h)
        \right)} \ge
    \frac{1}{p(n)}
  \end{align*}
for large enough $n\in \N$, where $h \leftarrow \cH_n$.
  
The generator $\Gc\colon \bit^c\times \bit^n \to \bit^m \times \bit^n$ is defined by
\begin{align*}
  \Gc(h, x) = (h(x), \; x).
\end{align*}

The public parameter length is $c$ (this is the description size of $h$), the generator consists of two blocks, and the maximal block length is $\max\{n,m\}$. Since the random coins of $\Gc$ define
$x$ and $x$ is completely revealed, the real Shannon entropy of $\Gc$ is $n$. That is,
\begin{align*}
  \EE{y\leftarrow G(U_c, U_n)}{\mathsf{RealH}_\Gc(y)} = n.
\end{align*}

Our goal in the remaining of this section is to show a non-trivial upper
bound on the accessible entropy of $\Gc$. We prove the following lemma.
\begin{lemma}\label{lemma:accessible}
There exists a polynomial $q(\cdot)$ such that for every $\Gc$-consistent online generator $\Gs$, it holds that
  \begin{align*}
  \EE{t\leftarrow
  T_{\Gs} (Z,1^n)}{\mathsf{AccH}_{\Gs}(t)} \le n-\frac{1}{q(n)}
\end{align*}
for all large enough $n\in \N$.
\end{lemma}

\begin{proof}
  Fix a $\Gc$-consistent online generator $\Gs$. Let us denote by $Y$ a
  random variable that corresponds to the first part of $\Gc$'s output (i.e., the first $m$
  bits) and by $X$ the second part (i.e., the last $n$ bits). Denote by $R$ the
  randomness used by the adversary to sample $Y$. Denote by $Z$ the random variable that corresponds to the description of the hash function $h$. Fix $q(n) \triangleq 4\cdot p(n)^2 $ Assume towards contradiction that for infinitely many $n$'s it holds that
  \begin{align*}
    \EE{t\leftarrow
    T_{\Gs} (Z, 1^n)}{\mathsf{AccH}_{\Gs}(t)} > n-\frac{1}{q(n)}.
  \end{align*}
  By Lemma~\ref{lemma:acc_entr}, this means that
  \begin{align}\label{eq:contradiction}
    \ShanEnt(Y\mid Z) + \ShanEnt(X \mid Y, Z, R) > n-\frac{1}{q(n)}
  \end{align}

  We show how to construct an adversary $\sA$ that can break the security of the dCRH.  The algorithm $\sA$, given a hash function $h\leftarrow \cH$, does the following:
  \begin{enumerate}
  \item Sample $r$ and let $y = \Gs(h, r)_1$
  \item Sample $r_1,r_2$ and output $x_1=\Gs(h, r,r_1)_2$ and
    $x_2 = \Gs(h, r,r_2)_2$.
  \end{enumerate}
  
  In other words, $\sA$ tries to create a collision by running $\Gc$ to get the
  first block, $y$, and then running it twice (by rewinding) to get two inputs
  $x_1,x_2$ that are mapped to $y$. Indeed, $\sA$ runs in
  polynomial-time and if $\Gs$ is
  $\Gc$-consistent, then $x_1$ and $x_2$ collide relative to $h$.  Denote by $Y^\sA$, $X_1^\sA$, and $X_2^\sA$ be random
  variables that correspond to the output of the emulated $\Gs$. Furthermore, denote by
  $(X_1^{\Col},X_2^{\Col})$ a random collision that $\Col(h)$ samples. To finish the
  proof it remains to show that
  \begin{align*}
    \EE{h\leftarrow \cH_n}{\SD((X_1^\sA,X_2^\sA), (X_1^{\Col},X_2^{\Col}))} \le \frac{1}{p(n)}
  \end{align*}
  which is a contradiction.

  By Pinsker's inequality (\Cref{prop:pinsker}) and the chain rule from
  \Cref{prop:kl_chain_rule}, it holds that
  \begin{align*}
    \SD&\left(\left(X_1^\sA,X_2^\sA\right), \left(X_1^{\Col},X_2^{\Col}\right)\right) \le \sqrt{\frac{\ln(2)}{2} \cdot \KL(X_1^\sA,X_2^\sA \|   X_1^{\Col},X_2^{\Col})} 
    \\  & = \sqrt{\KL\left(X_1^\sA \|X_1^{\Col}\right) + \EE{x_1\leftarrow
          X_1^\sA}{\KL(X_2^\sA|_{X_1^\sA=x_1} \| X_2^{\Col}|_{X_1^{\Col}=x_1})}}
          \\ & \leq \sqrt{\KL\left(X_1^\sA \|X_1^{\Col}\right)} + \sqrt{\EE{x_1\leftarrow
          X_1^\sA}{\KL(X_2^\sA|_{X_1^\sA=x_1} \| X_2^{\Col}|_{X_1^{\Col}=x_1})}}.
  \end{align*}
  
  Hence, by Jensen's inequality (\Cref{prop:jensen}), it holds that
\begin{align*}
    \EE{h\leftarrow \cH_n}{\SD((X_1^\sA,X_2^\sA), (X_1^{\Col},X_2^{\Col}))} \leq &  \sqrt{\EE{h\leftarrow \cH_n}{\KL(X_1^\sA \|X_1^{\Col})}} + \\ & \sqrt{\EE{\begin{subarray}{c}
h\leftarrow \cH_n\\
x_1\leftarrow X_1^\sA
\end{subarray}}{\KL(X_2^\sA|_{X_1^\sA=x_1} \|
      X_2^{\Col}|_{X_1^{\Col}=x_1})}}.
\end{align*}
  
  We complete the proof using the  following claims.
  \begin{myclaim} \label{claim:1} 
  It holds that 
  $$\EE{h\leftarrow \cH_n}{{\KL(X_1^\sA \|X_1^{\Col})}} \leq \frac{1}{p(n)^2}.$$
  \end{myclaim}
  \begin{myclaim} \label{claim:2}
  It holds that 
    $$\EE{\begin{subarray}{c}
h\leftarrow \cH_n\\
x_1\leftarrow X_1^\sA
\end{subarray}}{\KL(X_2^\sA|_{X_1^\sA=x_1} \|
      X_2^{\Col}|_{X_1^{\Col}=x_1})} \leq \frac{1}{p(n)^2}.$$
  \end{myclaim}

\begin{proof}[Proof of Claim~\ref{claim:1}]
    Recall that $X_1^{\Col}$ is the \emph{uniform} distribution over the inputs of
  the hash function and thus
  \begin{align*}
    \KL(X_1^\sA \|X_1^{\Col}) = \sum_x \pr{X_1^\sA = x}\cdot \log \frac{\pr{X_1^\sA
    = x}}{2^{-n}} = n - \ShanEnt(X_1^\sA).
  \end{align*}
  
  To sample $X_1^\sA$, the algorithm $\sA$ first runs $\Gs(r)_1$ to get $y$
  and then runs $G(r,r_1)$ to get $x_1$. Thus, by
  Equation~\eqref{eq:contradiction}, it holds that
  \begin{align*}
    \EE{h\leftarrow \cH_n}{\ShanEnt(X_1^\sA)} = \EE{h\leftarrow \cH_n}{\ShanEnt(X)}
    = \ShanEnt(X,Y\mid Z) = \ShanEnt(Y\mid Z) + \ShanEnt(X\mid Y,Z,R) \geq n-\frac{1}{q(n)},
  \end{align*}
  where the second equality follows since $\Gs$ is $\Gc$-consistent and thus
  $X$ fully determines $Y$. This implies that
  \begin{align*}
      \EE{h\leftarrow \cH_n}{{\KL(X_1^\sA \|X_1^{\Col})}} \leq  \frac{1}{q(n)} = \frac{1}{p(n)^2},
  \end{align*}
  as required.
\end{proof}

\begin{proof}[Proof of Claim~\ref{claim:2}]
  
  For $x_1 \in \supp(X_1^\sA)$, it holds that
  \begin{align*}
    \KL(X_2^\sA|_{X_1^\sA=x_1} \| X_2^{\Col}|_{X_1^{\Col}=x_1}) & = \sum_x \pr{X_2^\sA=x|_{X_1^\sA=x_1}}\cdot \log \frac{\pr{X_2^\sA=x|_{X_1^\sA=x_1}}}{|h^{-1}(h(x_1))|^{-1}}
    \\& = \log |h^{-1}(h(x_1))| - \ShanEnt(X_2^\sA |_{X_1^\sA=x_1}) .
  \end{align*}
  Hence,
  \begin{align*}
    \EE{\begin{subarray}{c}
h\leftarrow \cH_n\\
x_1\leftarrow X_1^\sA
\end{subarray}}{\KL(X_2^\sA|_{X_1^\sA=x_1} \|
    X_2^{\Col}|_{X_1^{\Col}=x_1})} & = \EE{\begin{subarray}{c}
h\leftarrow \cH_n\\
x_1\leftarrow X_1^\sA
\end{subarray}}{\log |h^{-1}(h(x_1))| -
                                \ShanEnt(X_2^\sA |_{X_1^\sA=x_1})} .
  \end{align*}
  
  Notice that the distribution of $X_2^\sA$ only depends on $y=h(x_1)$, that is,
  $X_2^\sA |_{X_1^\sA=x_1}$ is distributed exactly as $X_2^\sA |_{X_1^\sA=x_1'}$
  for every $x_1$ and $x_1'$ that such that $y=h(x_1)=h(x_1')$. Thus, we have
  that $X_2^\sA|_{X_1^\sA=x_1}$ is distributed exactly as $X|_{Y=y}$ and the
  distribution of $h(X_1)$ is distributed as $Y$. Namely,
  \begin{align*}
    \EE{\begin{subarray}{c}
h\leftarrow \cH_n\\
x_1\leftarrow X_1^\sA
\end{subarray}}{\KL(X_2^\sA|_{X_1^\sA=x_1} \|
    X_2^{\Col}|_{X_1^{\Col}=x_1})} & = \EE{\begin{subarray}{c}
h\leftarrow \cH_n\\
x_1\leftarrow X_1^\sA
\end{subarray}}{\log |h^{-1}(y)|} -
                                \EE{h\leftarrow \cH_n}{\ShanEnt(X \mid Y,R)}\\& = \EE{\begin{subarray}{c}
h\leftarrow \cH_n\\
x_1\leftarrow X_1^\sA
\end{subarray}}{\log |h^{-1}(y)|} -
                                \ShanEnt(X \mid Y, Z,R)
    \\ & \le \EE{\begin{subarray}{c}
h\leftarrow \cH_n\\
x_1\leftarrow X_1^\sA
\end{subarray}}{\log |h^{-1}(y)|}+\ShanEnt(Y\mid Z) -n +\frac{1}{q(n)} 
    \\ & =\frac{1}{q(n)},             
  \end{align*}
  where the first inequality follows by Equation~\eqref{eq:contradiction} and
  the second follows since
  \begin{align*}
   \EE{\begin{subarray}{c}
h\leftarrow \cH_n\\
y\leftarrow Y
\end{subarray}}{\log |h^{-1}(y)|}+\ShanEnt(Y\mid Z) & = \EE{\begin{subarray}{c}
h\leftarrow \cH_n\\
y\leftarrow Y
\end{subarray}}{\log
                                                        |h^{-1}(y)| + \ShanEnt_Y(y)}
    \\ & = \EE{\begin{subarray}{c}
h\leftarrow \cH_n\\
y\leftarrow Y
\end{subarray}}{\log \frac{|h^{-1}(y)|}{\pr{Y=y}}}
    \\ & \le \log \EE{\begin{subarray}{c}
h\leftarrow \cH_n\\
y\leftarrow Y
\end{subarray}}{\frac{|h^{-1}(y)|}{\pr{Y=y}}} = n,
  \end{align*}
  where the inequality is by Jensen's inequality (\Cref{prop:jensen}). Thus, overall
  \begin{align*}
      \EE{\begin{subarray}{c}
h\leftarrow \cH_n\\
x_1\leftarrow X_1^\sA
\end{subarray}}{\KL(X_2^\sA|_{X_1^\sA=x_1} \|
      X_2^{\Col}|_{X_1^{\Col}=x_1})} \leq {\frac{1}{q(n)}} = \frac{1}{p(n)^2},
      \end{align*}
    as required.
\end{proof}
\end{proof}

\subsection{From Statistically Hiding Commitments to \DCRH -- Proof of Theorem~\ref{thm:dcrh_from_stat_hiding}}\label{sec:dcrh_from_stat_hiding}

Let $\pi=(\sender, \receiver, \verifier)$ be a binding and statistically hiding  two-message commitment scheme. We show that there  exists a \DCRH family $\cH$.

 To sample a hash function in 
the family with security parameter $n$, we use the receiver's first 
message of the protocol. Namely, we set the hash function as $h \gets 
\receiver(1^n)$. Then, 
to evaluate $h$ on input $x$ we first parse $x$ as $x=(b,r)$, where $b$ is a bit, and output a commitment to the bit $b$ using randomness $r$, with respect to the receiver message $h$. That is, we set 
$$h(x)=\sender(h,b;r).$$

Since $\pi$ is efficient, then 
sampling and evaluating $h$ are polynomial-time procedures. This 
concludes the definition of our family $\cH$ of  hash functions. (Note that the functions in the family are not necessarily compressing.)

We next argue security. Suppose toward contradiction that $\cH$ is not a  \DCRH according to Definition~\ref{def:dcrh}. Then, for any $\delta(n)= n^{-O(1)}$ there exists an adversary $\sA$, such that
\begin{align}\label{eq:sdbreaker}
    \SD\left((h,\sA(1^n,h)),(h,\Col(h))\right) \le \delta,
  \end{align}
  for infinitely many $n$'s. From hereon, we fix $\delta$ to be any function such that $n^{-O(1)}< \delta< \frac{1}{2}-n^{-O(1)}$.

We show how to use $\sA$ to break the binding property of the commitment scheme. Our 
cheating receiver $\receiver^*$ is defined as follows: On input $h$,
$\receiver^*$ runs $\sA(h)$ to get $x$ and $x'$, interprets $x=(b,r)$ and $x'=(b',r')$ and outputs $b$ and $b'$ along with their openings $r$ and $r'$, respectively. Our goal is to show that $x=(b,r)$ and $x'=(b',r')$ are two valid distinct openings to the commitment scheme.

By Equation \eqref{eq:sdbreaker}, it suffices to analyze the success probability when the pair $(x,x')$ is sampled according to the distribution $\Col_h$, and show that it is at least $1/2-\negl(n)$. From the definition of $\Col_h$, we have that $h(x)=h(x')$ and thus 
$\sender(h,b;r)=\sender(h,b';r') \coloneqq y$. In other words, the second message of 
the protocol for $b$ with randomness $r$ and $b'$ with randomness 
$r'$ are the same, and thus both pass as valid 
openings in the reveal stage of the protocol: 
$\verifier(h,y,b,r)=1$ and $\verifier(h,y,b',r')=1$.

We are left to show that these are two \emph{distinct} openings for the 
commitment, namely, $b \neq b'$. To show this, we use the statistically hiding property of the commitment scheme. The following claim concludes the proof.
\begin{claim}
Fix any $h$. Then for $((b,r),(b',r')) \gets \Col(h)$ it holds that $\pr{b\neq b'} \ge 1/2 - \negl(n)~.$
\end{claim}
\begin{proof}
Let $B$ be the uniform distribution on bits and $R$ the uniform distribution on commitment randomness. For every commitment $c$, let $B_c$ be the distribution on bits given by sampling $(b,r) \gets (B,R)$ conditioned on $\sender(h,b;r)=c$. Let $C$ be the distribution on random commitments to a random bit.

By the statistical hiding property of the commitment scheme,
$$
\Delta((\sender(h,B,R),B),(\sender(h,B',R),B)) \leq \varepsilon\enspace,
$$
where $B'$ is an independent copy of $B$, and $\epsilon =\negl(n)$ is a negligible function. 
Furthermore,
$$
\Delta((\sender(h,B,R),B),(\sender(h,B',R),B)) =  \Delta((C,B_C),(C,B))=  \EE{\begin{subarray}{c} c \gets C\end{subarray}}{\Delta(B_c,B)} \enspace.
$$
By Markov's inequality, it holds that 
$$
\prob{c\gets C}{ \Delta(B_c,B) \geq \sqrt{\varepsilon} } \leq \sqrt{\varepsilon}\enspace.
$$
To conclude the proof note that
\begin{align*}
&\pr{b =  b' : (b,r),(b',r')\gets \Col_h} = \pr{b = b' \colon\;
\begin{array}{l}
(b,r) \gets (B,R)\\
c = \sender(h,b;r)\\
b' \gets B_c
\end{array}} \leq \\
&\pr{b = b' \colon 
\begin{array}{l}
(b,r) \gets (B,R)\\
c = \sender(h,b;r)\\
b' \gets B_c\\
\Delta(B_c,B) \leq \sqrt{\varepsilon}
\end{array}} + \prob{c\gets C}{
\Delta(B_c,B) \geq \sqrt{\varepsilon}}\leq\\
&\left(\frac{1}{2}+\sqrt{\varepsilon}\right)+\sqrt{\varepsilon}=\frac{1}{2}+\negl(n) ~.
\end{align*}
\end{proof}
Overall, the success probability of $\sA$ is at least $1/2 - \negl(n) -\delta\geq n^{-O(1)}$.

\paragraph{Using string commitments.}
The above proof constructs \DCRH from statistically hiding \emph{bit} commitment schemes. For schemes that support commitments to {\em strings}, following the above proof gives a stronger notion of \DCRH, where the adversary's output distribution is $(1-\negl(n))$-far from a random collision distribution. 

Technically, the change in the proof is to interpret $b$ in $x=(b,r)$ as a string of length $n$, rather than as a single bit.
The proof remains the same except that the probability that $b = b'$ is (negligibly close to) $2^{-n}$ instead of $1/2$. Thus, overall the success probability of $\sA$ is at least $1- \negl(n) - \delta$. To ensure a polynomial success probability we can allow any $\delta = 1- n^{-O(1)}$.

\section{From SZK-Hardness to Statistically Hiding Commitments}\label{sec:szk}
In this section, we give a direct construction of a constant-round statistically hiding commitment from average-case hardness in \SZK. This gives an alternative proof to \Cref{cor:szk}. 

\subsection{Hard on Average Promise Problems}

\begin{definition}
A promise problem $(\szkprob_Y,\szkprob_N)$ consists of two disjoint sets of {\em yes instances} $\szkprob_Y$ and {\em no instances} $\szkprob_N$.
\end{definition}

\begin{definition}
A promise problem $(\szkprob_Y,\szkprob_N)$ is hard on average if there exists a probabilistic polynomial-time sampler $\szkprob$ with support $\szkprob_Y\cup\szkprob_N$, such that for any probabilistic polynomial-time decider $\deci$, there exists a negligible function $\negl(n)$, such that

$$
\prob{r\gets \zo^n}{x \in \szkprob_{\deci(x)} \mid x\gets \szkprob(r)} \leq \frac{1}{2}+\negl(n)\enspace.
$$
\end{definition}

\subsection{Instance-Dependent Commitments}

\begin{definition}[\cite{OngV08}] 
An instance-dependent commitment scheme $\idc$ for a promise problem $(\szkprob_Y,\szkprob_N)$ is a commitment scheme where all algorithms get as auxiliary input an instance $x\in \zo^*$. The induced family of schemes $\set{\idc_x}_{x\in \zo^*}$ is
\begin{itemize}
\item 
statistically binding when $x\in \szkprob_N$,
\item 
statistically hiding when $x\in \szkprob_Y$.
\end{itemize} 
\end{definition}

\begin{theorem}[\cite{OngV08}]
Any promise problem $(\szkprob_Y,\szkprob_N)\in \SZK$ has a constant-round instance-dependent commitment.
\end{theorem}

\subsection{Witness-Indistinguishable Proofs}

\begin{definition}  
A proof system $\WI$ for an $\NP$ relation $R$ is witness indistinguishable if for any $x,w_0,w_1$ such that $(x,w_0),(x,w_1)\in R$, the verifier's view given a proof using $w_0$ is computationally indistinguishable from its view given a proof using $w_1$.
\end{definition}

Constant-round $\WI$ proofs systems are known from any constant-round statistically-binding commitments \cite{GMW87}. Statistically-binding commitments can be constructed from one-way functions \cite{Naor91}, and thus can also be obtained from average-case hardness in $\SZK$ \cite{OstrovskyW93}.
\begin{theorem}[\cite{GMW87,Naor91,OstrovskyW93}]
Assuming hard-on-average problems in $\SZK$, there exist constant-round witness-indistinguishable proof systems.
\end{theorem}

\subsection{The Commitment Protocol}
Here, we give the details of our protocol. Our protocol uses the following ingredients and notation:
\begin{itemize}
\item
 A $\WI$ proof for $\NP$.
\item
 A hard-on average $\SZK$ problem $(\szkprob_Y,\szkprob_N)$ with sampler $\szkprob$.
\item
 An instance-dependent commitment scheme $\idc$ for $\szkprob$.
\end{itemize}
We describe the commitment scheme in Figure \ref{fig:com_szk}.

\protocol
{Protocol \ref{fig:com_szk}}
{A constant round statistically hiding commitment from $\SZK$ hardness.}
{fig:com_szk}
{
Sender input: a bit $m \in \zo$.\\
Common input: security parameter $1^\secparam$.

\paragraph{Coin tossing into the well}
\begin{itemize}
\item
$\receiver$ samples $2\secparam$ independent random strings $\rho_{i,b}\gets \zo^\secparam$, for $i\in[\secparam],b\in \zo$. 
\item
The parties then execute (in parallel) $2\secparam$ statistically-binding commitment protocols $\sbc$ in which $\receiver$ commits to each of the strings $\rho_{i,b}$. We denote the transcript of each such commitment by $C_{i,b}$.

\item
$\sender$ samples $2\secparam$ independent random strings $ \sigma_{i,b}\gets \zo^\secparam$, and sends them to $\receiver$. 
\item
$\receiver$ sets $r_{i,b} = \rho_{i,b}\oplus\sigma_{i,b}$.
\end{itemize}

\paragraph{Generating hard instances}
\begin{itemize}
\item
$\receiver$ generates $2\secparam$ instances $x_{i,b} \gets \szkprob(r_{i,b})$, using the strings $r_{i,b}$ as randomness, and sends the instances to $\sender$.
\item
The parties then execute a $\WI$ protocol in which $\receiver$ proves to $\sender$ that there exists a $b\in \zo$ such that for all $i\in [\secparam]$, $x_{i,b}$ was generated consistently. That is, there exist strings $\set{\rho_{i,b}}_{i\in[\secparam]}$ that are consistent with the receiver's commitments $\set{C_{i,b}}_{i\in[\secparam]}$, and $x_{i,b}=\szkprob(\rho_{i,b}\oplus \sigma_{i,b})$.

As the witness, $\receiver$ uses $b=0$ and the strings $\set{\rho_{i,0}}_{i\in[\secparam]}$ sampled earlier in the protocol.
\end{itemize}

\paragraph{Instance-binding commitment}
\begin{itemize}
\item
The sender samples $2\secparam$ random bits $m_{i,b}$ subject to $m=\bigoplus_{i,b} m_{i,b}$.
\item
The parties then execute (in parallel) $2\secparam$ instance-dependent commitment protocols $\idc_{x_{i,b}}$ in which $\sender$ commits to each bit $m_{i,b}$ using the instance $x_{i,b}$.
\end{itemize}
}

\subsection{Analysis}

\begin{proposition}
Protocol \ref{fig:com_szk} is computationally binding.
\end{proposition}

\begin{proof}
Let $\sender^*$ be any probabilistic polynomial-time sender that breaks binding in Protocol~\ref{fig:com_szk} with probability~$\varepsilon$. We use $\sender^*$ to construct a probabilistic polynomial-time decider $\deci$ for the $\SZK$ problem $\szkprob$ with advantage $\varepsilon/4\secparam -\negl(n)$.

Given an instance $x \gets \szkprob$, the decider $\deci$ proceeds as follows:
\begin{itemize}
\item
It samples at random $i^*\in [\secparam]$ and $b^*\in \zo$.
\item
It executes the protocol $(\sender^*,\receiver)$ with the following exceptions:
\begin{itemize}
\item
The instance $x_{i^*,b^*}$, generated by $\receiver$, is replaced with the instance $x$, given to $D$ as input.
\item
In the $\WI$ protocol, as the witness we use $1 \oplus b^*$ and the strings $\set{\rho_{i,1\oplus b^*}}_{i\in[\secparam]}$  (instead of $0$ and the strings $\set{\rho_{i,0}}_{i\in[\secparam]})$.
\end{itemize}
\item
Then, at the opening phase, if $\sender^*$ equivocally  opens the $(i^*,b^*)$-th instance-dependent commitment, $D$ declares that $x\in \szkprob_Y$. Otherwise, it declares that $x\in \szkprob_\beta$ for a random $\beta\in\set{Y,N}$.
\end{itemize}

\paragraph{Analyzing $\deci$'s advantage.} Denote by $E$ the event that in the above experiment $\sender^*$ equivocally opens the $(i^*,b^*)$-th instance-dependent commitment. We first observe that the advantage of $D$ in deciding $\szkprob$ is at least as large as the probability that $E$ occurs.
\begin{myclaim}\label{clm:deci_prob} $\pr{ x \in \szkprob_{D(x)} } \geq \frac{1+\pr{E} }{2} - \negl(\secparam)$. 
\end{myclaim}
\begin{proof}By the definition of $D$,
\begin{align*}
&\pr{ x \in \szkprob_{D(x)} \mid E } =  
\pr{ x \in \szkprob_{Y} \mid E } = 1 - \pr{ x \in \szkprob_{N} \mid E } \geq 1 - \frac{\pr{E\mid x \in \szkprob_{N}}}{\pr{E}} 
\enspace,\\
&\pr{ x \in \szkprob_{D(x)} \mid \overline{E} } =  \frac{1}{2}\enspace.
\end{align*}
Furthermore, if $x\in \szkprob_N$ (namely, it is a no instance), then $\idc_x$ is binding, and thus
\begin{align*}
&\pr{E\mid x \in \szkprob_{N}}= \negl(\secparam)\enspace.
\end{align*}
Claim \ref{clm:deci_prob} now follows by the law of total probability.
\end{proof}

From hereon, we focus on showing that $E$ occurs with high probability.
\begin{myclaim}\label{clm:equi_prob}
$\pr{E}\geq \frac{\varepsilon}{2n} -\negl(\secparam)$.
\end{myclaim}

\begin{proof}
To prove the claim, we consider hybrid experiments $\hyb_0,\dots,\hyb_4$, and show that that the view of the sender $\sender^*$ changes in a computationally indistinguishable manner throughout the hybrids. We then bound the probability that $E$ occurs in the last hybrid experiment.
\begin{description}
\item[$\hyb_0$:] In this experiment, we consider an execution of $\deci(x)$ as specified above.
\item[$\hyb_1$:] Here $x$ is not sampled ahead of time, but rather first the value $\sigma_{i^*,b^*}$ is obtained from $\sender^*$, then a random value $\rho'\gets \zo^\secparam$ is sampled, and $x$ is sampled using randomness $r_{i^*,b^*}=\sigma_{i^*,b^*}\oplus \rho'$. Since $\rho'$ is sampled independently of the rest of the experiment, the sender's view in $\hyb_1$ is identically distributed to its view in $\hyb_0$. 
\item[$\hyb_2$:] Here the $(i^*,b^*)$-th commitment to $\rho_{i^*,b^*}$ is replaced with a commitment to $\rho'$. By the (computational) hiding of the commitment $\sbc$, the sender's view in $\hyb_2$ is computationally indistinguishable from its view in $\hyb_1$. 
\item[$\hyb_3$:] Here, in the $\WI$ protocol, instead of using as the witness $1 \oplus b^*$ and the strings $\set{\rho_{i,1\oplus b^*}}_i$, we use $0$ and the strings $\set{\rho_{i,0}}_i$. By the (computational) witness-indistinguishability of the protocol, the sender's view in $\hyb_3$ is computationally indistinguishable from its view in $\hyb_2$. 
\item[$\hyb_4$:] In this experiment, we consider a standard execution of the protocol between $\sender^*$ and $\receiver$ (without any exceptions). The sender's view in this hybrid is identical to its view in $\hyb_3$ (by renaming $\rho'=\rho_{i^*,b^*}$ and $x=x_{i^*,b^*}$).
\end{description}
It is left to bound from below the probability that $E$ occurs in $\hyb_4$. That is, when we consider a standard execution of $(\sender^*,\receiver)$ and sample $(i^*,b^*)$ independently at random.

Indeed, note that since the plaintext bit $m$ is uniquely determined by the bits $\set{m_{i,b}}_{i,b}$. Whenever $\sender^*$ equivocally opens the commitment to two distinct bits, there exists (at least one) $(i,b)$ such that $\sender^*$ equivocally opens the $(i,b)$-th instance-dependent commitment. Since in a standard execution $\sender^*$ equivocally opens the commitment with probability at least $\varepsilon$, and $(i^*,b^*)$ is sampled independently, $E$ occurs in this experiment with probability at least $\frac{\varepsilon}{2\secparam}$.

Claim \ref{clm:equi_prob} follows.
\end{proof}
This completes the proof that the scheme is binding.
\end{proof}

\begin{proposition}
Protocol \ref{fig:com_szk} is statistically hiding.
\end{proposition}

\begin{proof}
Let $\receiver^*$ be any (computationally unbounded) receiver. We show that the view of $\receiver^*$ given a commitment to $m=0$ is statistically indistinguishable from its view given a commitment to $m=1$.

For this purpose, consider the view of the receiver $\receiver^*$ after the coin tossing and instance-generation phase (and before the instance-dependent commitment phase). We shall refer to this as the {\em preamble view}. We say that the preamble view is {\em admissible}, if either of the following occurs:
\begin{itemize}
\item
Let $\set{x_{i,b}}_{i,b}$ be the instances sent by $\receiver^*$. Then there exists $i^*,b^*$ such that $x_{i^*,b^*} \in \szkprob_Y$.
\item
The sender $\sender$ rejects the $\WI$ proof that $\set{x_{i,b}}_{i,b}$ were properly generated.
\end{itemize}
To complete the proof, we show that the preamble view is admissible with overwhelming probability, and that conditioned on any admissible preamble view, the view of $\receiver^*$ given a commitment to $m=0$ is statistically indistinguishable from its view given a commitment to $m=1$. Since the preamble view is completely independent of $m$, the above two conditions are sufficient to establish  statistical indistinguishability of the total views.

\begin{myclaim}\label{clm:adm_prob}
The probability that the preamble view is not admissible is negligible.
\end{myclaim}
\begin{proof}
Let $A$ be the event that the $\WI$ proof is accepted and let $Y$ be the event that for some $(i,b)$, $x_{i,b}$ is a yes instance. To show that the preamble view is not admissible with negligible probability, we would like to prove that
$$
\pr{A\wedge \overline{Y}} \leq \negl(\secparam)\enspace.
$$

Let $T$ be the event that the statement proven by $\receiver^*$ in the $\WI$ protocol is true. Namely, there exists $b\in \zo$ such that all $\set{x_{i,b}}_{i}$ are generated consistently with the coin-tossing phase (and in particular where the coin-tossing phase consists of valid commitments $\set{C_{i,b}}_{i}$).

First, note that by the soundness of the $\WI$ system, the probability that the preamble is admissible, and in particular the proof is accepted, when the statement is false, is negligible:
$$\pr{A \wedge \overline{T}}\leq \negl(\secparam)\enspace.$$
We now show:
$$\pr{\overline{Y} \wedge {T}}\leq \negl(\secparam)\enspace.$$

For this purpose, fix any $\sbc$ commitments $\set{C_{i,b}}_{i,b}$. Let $F=F[\set{C_{i,b}}_{i,b}]$ be the event, over the sender randomness $\set{\sigma_{i,b}}_{i,b}$, that there exists $\beta\in \zo$ such that $\set{C_{i,\beta}}_i$ are valid commitments to strings $\set{\rho_{i,\beta}}_i$ and for all $i$, $\szkprob(\rho_{i,\beta}\oplus \sigma_{i,\beta})=x_{i,\beta}\in \szkprob_N$. We show
$$
 \pr{F} \leq 2^{-\Omega(\secparam)}\enspace.
$$
This is sufficient since 
$$
\pr{\overline{Y}\wedge T} \leq \max_{\begin{subarray}{c}
C_{1,0} \dots  C_{\secparam,0}\\
C_{1,1} \dots  C_{\secparam,1}
\end{subarray} }\pr{F} \leq 2^{-\Omega(\secparam)}\enspace.
$$
To bound the probability that $F$ occurs, fix any $\beta$ and commitments $\set{C_{i,\beta}}_i$ to strings $\set{\rho_{i,\beta}}_i$. Then the strings $\rho_{i,\beta}\oplus\sigma_{i,\beta}$ are distributed uniformly and independently at random. Since $\szkprob \in \szkprob_Y$ with probability at least $0.49$, and taking a union bound over both $\beta\in \zo$, the bound follows.

This concludes the proof of Claim \ref{clm:adm_prob}.
\end{proof}

\begin{myclaim}\label{clm:stat_ind}
Fix any admissible preamble view $V$. Then, conditioned on $V$ the view of $\receiver^*$ when given a commitment to $m=0$ is statistically indistinguishable from its view when given a commitment to $m=1$.  
\end{myclaim}

\begin{proof}
If $V$ is such that the $\WI$ proof is rejected then $\sender$ aborts and the view of $\receiver^*$ remains independent of $m$. Thus, from hereon, we assume that the instances corresponding to $V$ include an instance $x_{i^*,b^*}\in \szkprob_Y$. In particular, the corresponding instance-dependent commitment $\idc_{x_{i^*,b^*}}$ is statistically hiding.

It is left to note that in any execution $(\sender,\receiver^*)$, with either $m\in\zo$, the bits $M_{-i}:=\set{m_{i,b}}_{(i,b)\neq (i^*,b^*)}$ are distributed uniformly and independently at random. Conditioned on $V$ and $M_{-i}$, only the bit 
$$m_{i^*,b^*} = m\bigoplus_{m'\in M_{-i}} m' $$ 
depends on $m$. By the statistical hiding of $\idc_{x_{i^*,b^*}}$ a commitment to $0\bigoplus_{m'\in M_{-i}} m'$ is statistically indistinguishable from a commitment to $1\bigoplus_{m'\in M_{-i}} m'$.

This concludes the proof of Claim \ref{clm:stat_ind}.
\end{proof}

\end{proof}

\bibliographystyle{alpha}
\bibliography{crypto}

\end{document}